\documentclass[a4paper,UKenglish,cleveref, autoref, thm-restate]{lipics-v2021}

\pdfoutput=1 
\hideLIPIcs  
\nolinenumbers

\usepackage{xspace, amssymb,xcolor, marginnote, mathtools}
\globalcolorstrue
\usepackage{tikz}
\usetikzlibrary{arrows,decorations,calligraphy,calc}
\graphicspath{{./figures/}}

\title{Decision problems on geometric tilings} 
\subtitle{Wang potatoes and Finite Local Complexity}

\titlerunning{Decision problems on geometric tilings} 

\author{Benjamin Hellouin de Menibus}{Université Paris-Saclay, CNRS, Laboratoire Interdisciplinaire des Sciences du Numérique, 91400, Orsay, France \and \url{https://www.lisn.fr/~hellouin/}}{hellouin@lisn.fr}{https://orcid.org/0000-0001-5194-929X}{}

\author{Victor H. Lutfalla}{Aix-Marseille Univ, CNRS, I2M, Marseille, France \and \url{https://www.lutfalla.fr/}
}{victor.lutfalla@math.cnrs.fr}{https://orcid.org/0000-0002-1261-0661}{ANR JCJC 2019 19-CE48-0007-01}

\author{Pascal Vanier}{Normandie Univ, UNICAEN, ENSICAEN, CNRS, GREYC, 14000
  128   Caen, France\and
  \url{https://vanier.users.greyc.fr/}}{pascal.vanier@unicaen.fr}{https://orcid.org/0000-0001-9207-9112}{ANR JCJC 2019 19-CE48-0007-01}

\authorrunning{B. Hellouin de Menibus, V. Lutfalla, P. Vanier}

\Copyright{Benjamin Hellouin de Menibus\and Victor Lutfalla \and Pascal Vanier}



\category{} 

\relatedversion{} 

\EventEditors{John Q. Open and Joan R. Access}
\EventNoEds{2}
\EventLongTitle{42nd Conference on Very Important Topics (CVIT 2016)}
\EventShortTitle{CVIT 2016}
\EventAcronym{CVIT}
\EventYear{2016}
\EventDate{December 24--27, 2016}
\EventLocation{Little Whinging, United Kingdom}
\EventLogo{}
\SeriesVolume{42}
\ArticleNo{23}

\newcommand{\NN}{\ensuremath{\mathbb{N}}\xspace}
\newcommand{\ZZ}{\ensuremath{\mathbb{Z}}\xspace}
\newcommand{\RR}{\ensuremath{\mathbb{R}}\xspace}
\newcommand{\norm}[1]{\left\lVert#1\right\rVert}
\newcommand{\manyinf}{\ensuremath{\leq_m}\xspace}
\newcommand{\sizu}{\ensuremath{\Sigma^0_1}\xspace}
\newcommand{\pizu}{\ensuremath{\Pi^0_1}\xspace}
\newcommand{\haltingprob}{\ensuremath{\mathbf{HP}}\xspace}

\newcommand{\shape}{\ensuremath{t}\xspace}
\newcommand{\shapeset}{\ensuremath{\mathcal{T}}\xspace} 
\newcommand{\tile}{\ensuremath{\mathbf{t}}\xspace}
\newcommand{\tiling}{\ensuremath{T}\xspace}
\newcommand{\tileset}{\ensuremath{\mathbf{T}}\xspace}

\newcommand{\pattern}{\ensuremath{\mathbf{P}}\xspace}
\newcommand{\support}[1]{\ensuremath{\mathrm{supp}(#1)}\xspace}
\newcommand{\appears}{\sqsubseteq}
\newcommand{\avoids}{\not\sqsubseteq}
\newcommand{\locallyallowed}{\ensuremath{\mathcal{L}}\xspace} 
\newcommand{\domino}{\textsc{domino}}
 
\newcommand{\subshift}{\ensuremath{X}\xspace}
\newcommand{\patternsize}{\ensuremath{D_{max}}\xspace} 
\newcommand{\position}{\ensuremath{f}\xspace} 

\begin{document}

\maketitle

\begin{abstract}
We study decision problems on geometric tilings. First, we study a variant of the Domino problem where square tiles are replaced by geometric tiles of arbitrary shape. We show that this variant is undecidable regardless of the shapes, extending the results of \cite{hellouin2023} on rhombus tiles. This result holds even when the geometric tiling is forced to belong to a fixed set. Second, we consider the problem of deciding whether a geometric subshift has finite local complexity, which is a common assumption when studying geometric tilings. We show that this problem is undecidable even in a simple setting (square shapes with small modifications).
\end{abstract}

\paragraph*{Disclaimer} The figures in this article are somewhat optimised for being viewed on a computer.

\section{Introduction}

There are two main settings when talking about tilings: geometric tilings and symbolic tilings. Geometric tilings are coverings of the plane by arbitrary shapes, while symbolic tilings are tilings of $\ZZ^2$ by coloured unit squares where some forbidden patterns are not allowed to appear. One main fundamental difference between the two settings is that geometric tilings have no rules outside of the geometric constraints given by the shapes themselves, while in the symbolic case the corners of the unit squares must be on integer coordinates. 

In this article, we study two decisions problems concerning geometric tilings: the domino problem in a geometric tiling, and the problem of finite local complexity; we prove that both problems are undecidable.

The undecidability of the domino problem \cite{berger1966} is a founding result of the theory of symbolic tilings: given a set of forbidden patterns, is it possible to tile the plane while avoiding this set of forbidden pattern? The domino problem has subsequently been studied on many structures: many groups \cite{aubrun2018}, the hyperbolic plane \cite{margenstern2008, kari2008, goodman-strauss2010}, tilings by rhombuses by the first two authors \cite{hellouin2023}, other symbolic tilings \cite{aubrun2020}, self-similar fractal structures \cite{barbieri2016}, etc. In this article we focus on a fusion of the two models, \emph{symbolic-geometric tilings}: they are geometric tilings with \emph{colored local rules}, that consist in adding decorations and a finite number of forbidden patterns on these decorations (precise definitions can be found in \autoref{defn:symbgeom}). Colored local rules have been used to capture properties of geometric tilings, such as when they are generated by a substitution \cite{goodman1999} or planar \cite{fernique2019}.

Our first problem of interest is the domino problem, as studied in \cite{hellouin2023} for rhombus tilings: fixing a nonempty set of geometric tilings $X$, given as input some colored local rules on the shapes of $X$, can we tile the plane with shapes of $X$ decorated in such a way that they avoid all forbidden patterns? We show that this problem is always undecidable when the shapes are compact and simply connected. The proof can be seen as representing a quasi-isometry between $\mathbb Z^2$ and the adjacency graph of a geometrical tiling, in order to reduce the domino problem on $\mathbb Z^2$ which is known to be undecidable. We discuss the role of quasi-isometries and related results in Section~\ref{sec:quasi-isometry}.

In geometric tilings there are potentially infinitely many ways to put two shapes side by side inside a valid tiling. This is the notion of \emph{local complexity}: how many patterns may appear in a ball of some radius up to translation. In symbolic tilings, this complexity is always finite. Geometric tilings sharing this property are said to have \emph{finite local complexity} or \emph{FLC} for short. FLC is a fundamental property in the study of geometric tilings and has been introduced by Kenyon~\cite{kenyon1992}\footnote{Under the terminology \emph{finite number of local pictures}.}.

Most results about geometric tilings are under the assumption of FLC and most well known examples such as Penrose tilings or the hat \cite{smith2024} have FLC, although tilings without FLC have also been studied (see e.g. \cite{NPB2015}). FLC makes geometric tilings behave similarly to symbolic tilings, in the sense that locally valid patterns may be enumerated by an algorithm.

The second problem we study is the following: given a geometric tileset, is it decidable whether it has FLC? We show that having FLC is a $\Sigma^0_1$-complete problem and is thus undecidable. 

The paper is organised as follows. After providing the relevant definitions in \autoref{sec:defs}, we prove the undecidability of the domino problem for symbolic geometric tilings in \autoref{sec:dominoUndec} and the undecidability of FLC for geometric tilings in \autoref{sec:undecFLC}. 

\section{Definitions and settings }\label{sec:defs}

\subsection{Computability}
A decision problem $A: \NN\subseteq \{0,1\}$ is said to be \emph{computable} when there exists a Turing machine $M$ that halts on all inputs $x\in \NN$, and which outputs $1$ iff $A(x) = 1$. The canonical undecidable problem is the \emph{halting problem}, denoted \haltingprob: given (the number corresponding\footnote{Using the usual enumeration of Turing machines.} to) a Turing machine as an input, does it halt on an empty input tape? 

The halting problem is \emph{recursively enumerable}: there is a Turing machine which never halts and outputs every input that maps to $1$ (in this case, every Turing machine that halts). Moreover, it is complete for the class of recursively enumerable problems under \emph{many-one reduction}, defined as follows: we say that $A \manyinf B$ if there exists a total computable function $f$ such that $B(f(x)) = 1 \Leftrightarrow A(x) = 1$. The completeness of the halting problem means that for any recursively enumerable problem $A$, $A\manyinf \haltingprob$.

We denote the class of r.e. and co-r.e problems as $\sizu$ and $\pizu$, respectively, since they are the first levels of the arithmetical hierarchy. This means that \haltingprob is $\sizu$-complete and its complement $\overline{\haltingprob}$ is $\pizu$-complete.

\subsection{Symbolic tilings}
Given a finite alphabet $\Sigma$, a \emph{configuration} $c$ is an assignment of letters to each coordinate of the plane $\ZZ^2$: $c : \ZZ^2 \to \Sigma$. A \emph{pattern} $p : D \to \Sigma$ is an assignment of letters to a finite portion of the plane $D\Subset \ZZ^2$. The shift of a configuration $c$ by a vector $v\in\ZZ^2$ is denoted $\sigma_v(c)$: $\sigma_v(c)(z)=c(z+v)$. A pattern $p : D\to\Sigma$ \emph{appears} in a configuration $c$, denoted $p\appears c$, if there exists $z\in\ZZ^2$ such that $\sigma_z(c)_{|D}=p$. Conversely, we denote $p \avoids c$ when a pattern does not appear in a configuration.

We define a distance between configurations, two configurations $x$ and $y$ being near each other when they coincide on a big central ball:
\[ d(x,y) = 2^{-min\{\norm{z} \mid z\in\ZZ^2,~ x(z)\neq y(z) \}}\text{.}\]
This defines a topology on the set of all configurations $\Sigma^{\ZZ^2}$. A \emph{subshift} is a closed, shift-invariant subset of $\Sigma^{\ZZ^2}$.
Subshifts are exactly the sets of configurations which can be defined by forbidding sets of patterns. That is to say, for each subshift $X\subseteq \Sigma^{\ZZ^2}$, there exists a family of forbidden patterns $F$ such that
\[
X=\left\{ x \in \Sigma^{\ZZ^2}\mid \forall p \in F, p\avoids x \right\}
\]
A \emph{subshift of finite type} or \emph{SFT} is a subshift for which there exists a finite family of forbidden patterns defining it. 

\begin{definition}[Classical domino problem \cite{kahr1962}]
  The domino problem denoted \domino{} is the following problem: given as input be an SFT $X$ given by an alphabet $A$ and a family of forbidden patterns $F$, is $X$ non-empty?
\end{definition}

\begin{theorem}[\cite{berger1966}]
  The classical domino problem is undecidable and $\Pi_1^0$-complete.
\end{theorem}

It is well-known that every SFT $X$ is conjugate to another SFT $\phi(X)$ defined by nearest-neighbour forbidden patterns (i.e. connected support of size $2$), and $\phi(X)$ is computable from $X$.
As a direct corollary, the classical nearest-neighbour domino problem is also undecidable and $\Pi_1^0$-complete.

\subsection{Geometric tilings}

\begin{definition}[Geometric tilings]
  A shape $\shape$ is a closed topological disk (compact and simply connected).
  A shapeset $\shapeset$ is a finite set of shapes.
  A tiling $\tiling$ is a covering of the plane by countably many shapes with disjoint interior.
  A $\shapeset$-tiling is a tiling where all shapes are in $\shapeset$ up to translation.
\end{definition}

As a running example we consider the square-triangle shapeset $\shapeset_\triangle$ containing the square (in 3 orientations) and the equilateral triangle (in 4 orientations), illustrated on Fig~\ref{fig:shapeset}.
For technical reasons detailed in Section \ref{sec:dominoUndec} we consider all sides equal to $\tfrac{5}{2}$.

\begin{figure}[htp]
  \begin{center}
    \includegraphics[width=0.4\textwidth]{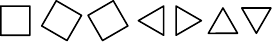}
  \end{center}
  \caption{The square-triangle shapeset $\shapeset_\triangle$.}
  \label{fig:shapeset}
\end{figure}

\begin{remark}[Shapes and tiles]
  The usual terminology is tiles rather than shapes.
  However we chose to differentiate between shapes (purely geometric) and tiles (both geometric and symbolic) that we define below.
\end{remark}

\begin{remark}[Shapes and polygons]
  As argued in \cite{kenyon1992}, any geometric tiling by topological disks is combinatorially equivalent to a polygonal tiling (where all shapes are polygons), though some properties of the tilings such as self-similarity might be lost in this equivalence. We write our results for topological disks as it does not make the proofs significantly more complicated; the next proposition is the only technical property that we use.
\end{remark}

\begin{proposition}\label{prop:countable-segments}
The boundary of a shape $t$ contains at most countably many maximal segments (that are not included in a larger segment).
\end{proposition}

\begin{proof}
A shape $t$ is, by assumption, homeomorphic to a unit disk. The images under this homeomorphism of  maximal segments in the boundary of $t$ are disjoint connected subsets of the unit circle, that is, intervals, that each contain some point of the form $e^{ri\pi}$ for $r\in \mathbb Q$; therefore there cannot be more than countably many of them.
\end{proof}

\begin{definition}[Pattern and support]
  We call \emph{pattern} $\pattern$ a finite and simply connected set of non-overlapping shapes.
  The \emph{support} of $\pattern$, denoted $\support{\pattern}$ is the union of its shapes. We usually consider patterns up to translation, and the context should make it clear.
\end{definition}

We write $\pattern \in \tiling$ when pattern $\pattern$ is a subset of a tiling $\tiling$, and $\pattern \appears \tiling$ (resp. $\pattern \appears X$) when a pattern appears up to translation in $T$ (resp.~in some tiling $\tiling \in X$).

We choose a distance between geometric tilings that generalises the usual distance on $\mathbb{Z}^2$ configurations. Two tilings $\tiling_1$ and $\tiling_2$ are $\varepsilon$-close if they are identical on a disk of radius $\varepsilon^{-1}$ up to a translation of norm at most $\varepsilon$. See \cite{robinson2004, sadun2008, lutfalla2022} for full details.

With the topology induced by this distance we generalise the concept of subshift to geometric subshifts as follows.

\begin{definition}[Geometric subshifts]
  Given a finite set of shapes $\shapeset$, a \emph{geometric subshift} $X$ on $\shapeset$ is a set of $\shapeset$-tilings that is invariant under $\mathbb{R}^2$-translations and closed for the geometric tiling topology.
\end{definition}

In particular, a geometric subshift is purely geometric if it contains exactly all possible $\shapeset$-tilings for some shapeset $\shapeset$; in other words, there are no other constraints than those imposed by the geometry of $\shapeset$. Any geometric subshift is a subset of the corresponding purely geometric subshift that avoids some set of forbidden patterns:

\begin{definition}[Valid tiling] 
  Given a finite set of patterns $G$, we say that a tiling $\tiling$ is \emph{valid} for $G$ or \emph{avoids} $G$ when for any $\pattern \in G$, $\pattern \avoids \tiling$.
\end{definition}




\begin{definition}[FLC]\label{def:FLC}
  A set of geometric tilings $X$ has finite local complexity (FLC) when for any compact $K\subset \mathbb{R}^2$, up to translation, there exist finitely many patterns $\pattern \appears X$  such that $\support{\pattern}\subset K$.
  Equivalently, there are finitely many 2-tile patterns in $X$. 
\end{definition}

Typical examples of geometric subshifts with FLC are the polygonal edge-to-edge tilings. For example, the purely geometric subshift given by $\shapeset_\triangle$ does not have FLC (even a tiling using only a single square tile can have lines or columns shifted by arbitrary offsets), but restricting to edge-to-edge tilings gives it FLC. Continuing, the subshift $X_\triangle$ of edge-to-edge tilings containing no bars (3 consecutive squares) and no tribars (4 unit triangles aligned as in the triangular grid) also has FLC; see Fig.~\ref{fig:forbbiden_patterns}. 

\begin{figure}[htp]
  \center
    \includegraphics[width=0.3\textwidth]{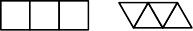}
  \caption{The two forbbiden patterns of $X_\triangle$ up to rotation and reflexion. There are in total $12$ forbidden patterns up to translation.}
  \label{fig:forbbiden_patterns}
\end{figure}

\begin{figure}[htp]
  \begin{center}
    \includegraphics[width=0.8\textwidth]{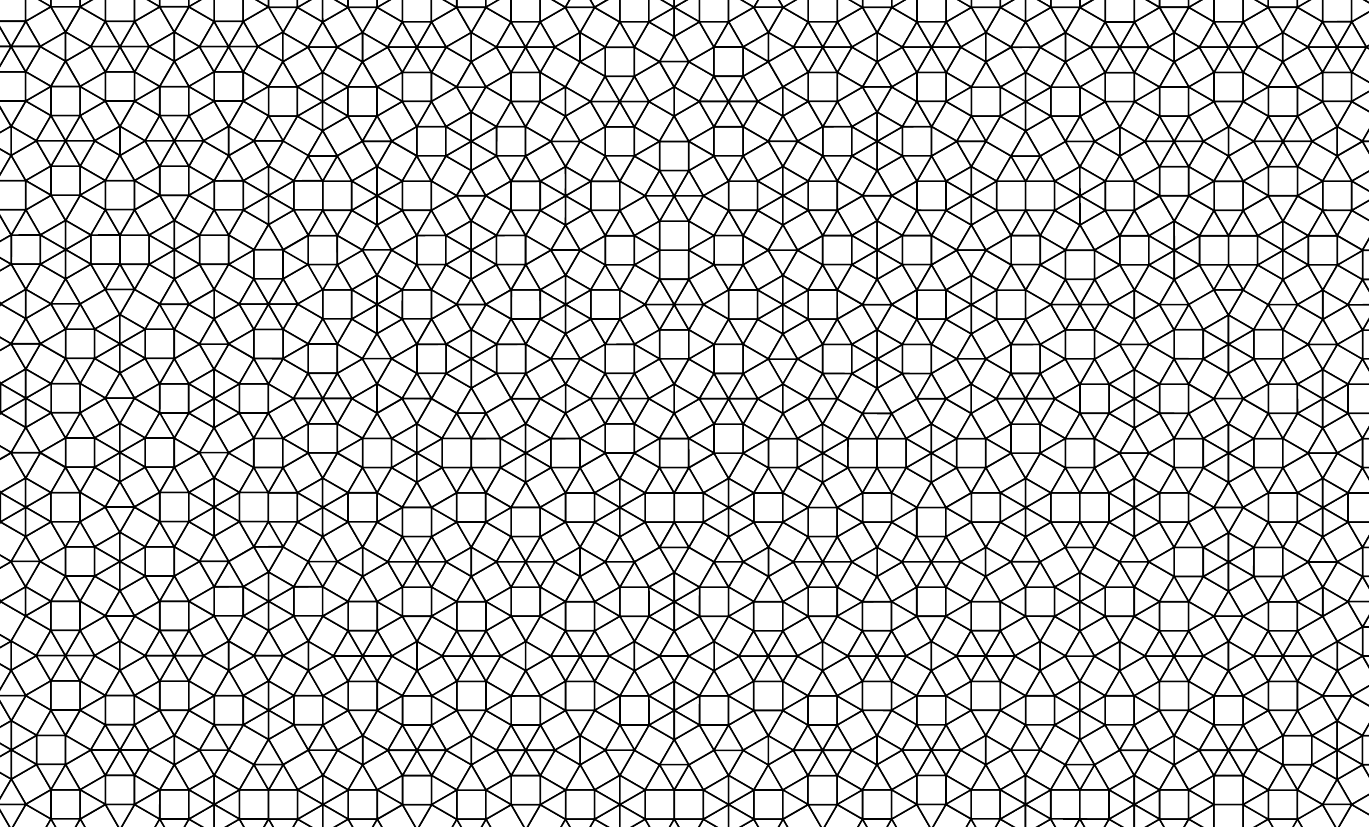}
  \end{center}
  \caption{A typical fragment of a tiling in $X_\triangle$.}
  \label{fig:square-triangle_tiling}
\end{figure}

\subsection{Symbolic-geometric tiling and the domino problem}\label{defn:symbgeom}

\begin{definition}[Symbolic-geometric tilings]
  Given a finite alphabet $A$ and a finite shapeset $\shapeset$, a \emph{symbolic-geometric tile} $\tile$ is a shape $\shape$ with a label $a$, \emph{i.e.}, $\tile = (\shape, a)\in \shapeset\times A$.
  We define the \emph{label erasor} or \emph{projection} $\pi$ as  $\pi((\shape,a)) \coloneq \shape$.
  We call \emph{symbolic-geometric tileset} a subset $\tileset$ of $\shapeset\times A$.
  We call \emph{symbolic-geometric tiling} $\tiling$ a covering of $\RR^2$ by non-overlapping symbolic-geometric tiles, that is $\pi(\tiling) \coloneq \{ \pi(\tile), \tile \in \tiling\}$ is a geometric tiling.
\end{definition}

We generalise as expected the notions of symbolic-geometric patterns $\pattern$ and $\tileset$-tiling.

\begin{definition}[Domino on $X$]
  Given a finite shapeset $\shapeset$ and a set of $\shapeset$-tilings $X$, the problem $\domino_X$ is defined as:
  \begin{description}
  \item[Input:] a finite alphabet $A$, a tileset $\tileset\subset \shapeset \times A$ and a finite set $G$ of forbidden $\tileset$-patterns.
  \item[Output:] does there exist a $\tileset$-tiling $\tiling$ that is valid for $G$ and such that $\pi(\tiling)\in X$?
  \end{description}
\end{definition}

Forbidden patterns on decorated shapes have also been called \emph{colored local rules} in the litterature. The case where $\shapeset$ contains only a unit square and $X$ is the subshift of edge-to-edge square tilings (containing a single point up to translation) corresponds to the classical domino problem. This definition, as well as an equivalent of the following result, appeared in \cite{hellouin2023} in the particular case of rhombus shapes.

\section{The domino problem is undecidable on every set of tilings}\label{sec:dominoUndec}

\begin{theorem}[Main result]\label{thm:main}
  Let $X$ be a non-empty set of tilings with finite local complexity.
  The domino problem on $X$ is $\Pi_1^0$-hard and therefore undecidable.
\end{theorem}

This generalises the classical domino problem (when $X$ is a single edge-to-edge tiling using a single square shape) as well as Theorem 1 from \cite{hellouin2023} for geometric subshifts on rhombus shapes. However, the proof in \emph{op.~cit.} made heavy use of the particular structure of rhombus tilings, while the constructions in the present paper are much more generic.

We proceed by reduction to the classical nearest neighbour domino problem.
Let $(A,F)$ be an input, that is, $A$ is a finite alphabet and $F$ is a finite set of nearest-neighbour forbidden $\mathbb{Z}^2$ patterns. We build a tileset $\tileset$ on shapes $\shapeset$ and a set of forbidden patterns $G$ such that there exists a $\tileset$-tiling $T$ valid for $G$ with $\pi(T)\in X$ if, and only if, the SFT $(A,F)$ is not empty.

We illustrate our construction throughout on the classical quarter-plane subshift (see Fig.~\ref{fig:quarter-plane}) and the geometric subshift $X_\triangle$, where the edge-to-edge condition implicitly induces the list $\locallyallowed$ of allowed 2-shapes patterns. 
\begin{figure}[htp]
  \center

  \includegraphics[width=0.6\textwidth]{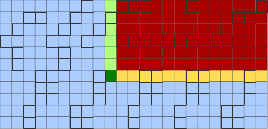}

  \caption{A fragment of a typical tiling of the nearest-neighbour quarter-plane subshift on blue and red tiles (with green and yellow boundary tiles).\\
    Any two-tiles pattern not appearing in this fragment is forbidden.}
  \label{fig:quarter-plane}
\end{figure}

\subsection{Coloured tiles}\label{sec:placeholder}

We assume (by rescaling if necessary) that each shape in $\shapeset$ strictly contains a unit square and does not have in its boundary two vertical or horizontal segments at integer distance from each other, using Proposition~\ref{prop:countable-segments}.
For each shape, we pick an arbitrary origin point from which we draw a square grid of width $\tfrac{1}{4}$, as illustrated in Fig.~\ref{fig:quadrillage} (we pick a vertex as origin point in the figures).
We call \emph{checkered tiles} these shapes with the superimposed $\tfrac{1}{4}$-grid.

The choice to superimpose a grid of width $\tfrac{1}{4}$, together with the choice of coding cells below, ensures that: any labelled tile contains a coding cell (Fig.~\ref{fig:placeholder_tiles}), adjacent tiles have coherent colourings of their common coding cells (Fig.~\ref{fig:forbidden_GF}), and in any tiling the coding cells are arranged on a valid $\mathbb{Z}^2$ grid (Section~\ref{sec:grid_structure}).
This grid could be replaced with any finer grid (for example of width $\tfrac{1}{k}$ with $k\geq 4$): in that case we might change the choice of coding cells.

\begin{figure}[htp]
  \center
  \includegraphics[width=0.8\textwidth]{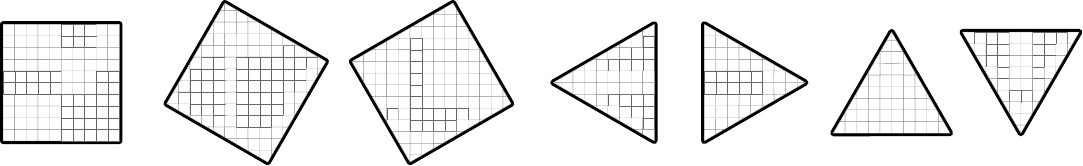}
  \caption{The $\tfrac{1}{4}$ square grid superimposed on the square and triangle shape of $X_\triangle$ forming the checkered square-triangle tiles. Recal that we defined the square and triangle shapes with sides $\tfrac{5}{2}$.}
  \label{fig:quadrillage}
\end{figure}

For every $t\in \shapeset$, we define $16$ placeholder tiles, one for each $0\leq i,j \leq 3$, as shown in Fig.~\ref{fig:placeholder_tiles}. In the $(i,j)$-placeholder tile, the cell $(i',j')$ is a \emph{coding cell} if $i'=i\bmod 4$ and $j'=j\bmod 4$ or a \emph{linking cell} if either $i'=i\bmod 4$ or $j'=j\bmod 4$. A cell such that $i'=i\bmod 4$ is also a row cell, and a cell such that $i'=i\bmod 4$ is also a column cell; note that a coding cell is both row and column.

We define $\tileset_0$ the set of placeholder tiles.

\begin{figure}[htp]
  \center
  \includegraphics[width=\textwidth]{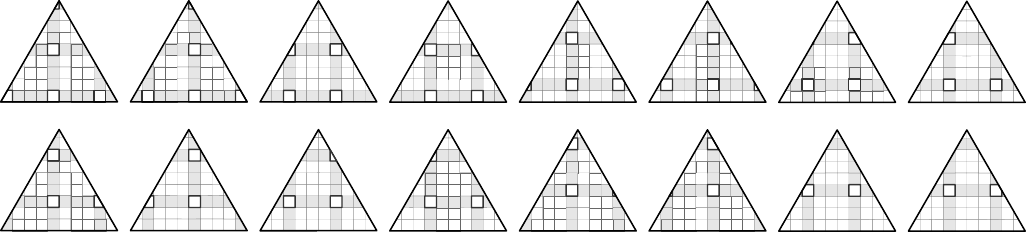}
  \caption{The $16$ placeholder tiles for a triangle shape. The coding cells are shown with a bold boundary and the linking cells in gray.}
  \label{fig:placeholder_tiles}
\end{figure}

We now decorate each placeholder tile $\tile_0$ in $\tileset_0$, using the alphabet of the SFT $(A,F)$: we define a coloured tile $\tile$ for each possible $A$-labelling of the coding cells. We denote by $\tileset$ the set of coloured tiles obtained from this process. As the shapeset is finite, the diameter of the shapes is bounded and therefore they contain a bounded number of coding cells; assuming each tile in $\tileset_0$ contains at most $n$ coding cells, it defines at most $|A|^n$ coloured tiles, so $\tileset$ is finite. 

\begin{figure}[htp]
  \center
  \includegraphics[width=\textwidth]{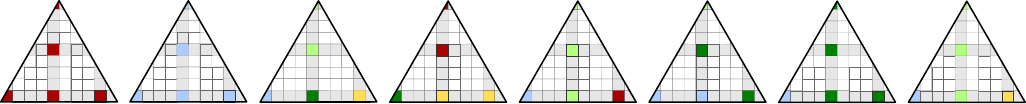}
  \caption{8 colourings of the first placeholder tile of Fig.~\ref{fig:placeholder_tiles}.\\
    As there are 5 coding cells and 5 colours, the total number of coloured tiles for this placeholder is $5^5$, most of which are later forbidden by $G_F$. The five leftmost tiles are allowed and the three rightmost tiles are later forbidden as they contain a forbidden pattern or force a forbidden pattern.}
  \label{fig:coloured_tiles}
\end{figure}

\begin{figure}[htp]
  \center
  \includegraphics[width=0.6\textwidth]{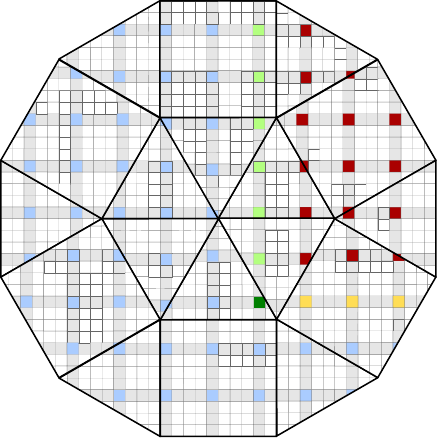}
  \caption{A typical allowed pattern of coloured tiles.}
  \label{fig:coloured_pattern}
\end{figure}

\subsection{Forbidden patterns}
We define a set of forbidden patterns $G = G_{grid}\cup G_{coding} \cup G_{F}$ that enforces a global grid structure among coding cells of all coloured tiles. We forbid all coloured patterns that break the grid structure ($G_{grid}$), do not encode a $\mathbb Z^2$ tiling in the coding cells ($G_{coding}$), or induce a forbidden pattern for $F$ in the coding cells ($G_F$). Figure~\ref{fig:coloured_pattern} illustrates a typical allowed pattern.

As the set $X$ has FLC, there is a finite list $\locallyallowed$ of allowed $2$-shapes patterns. That is, any pattern $\pattern$ of two adjacent tiles appearing in any tiling $\tiling \in X$ is in $\locallyallowed$ up to translation.
Notice that the set $X$ (or its tileset) is not a part of the input, so we do not require the list $\locallyallowed$ to be computable from $X$.

\subsubsection{Support of the forbidden patterns}

\begin{definition}[$r$-patterns]
For $r>0$, a \emph{$r$-pattern} $P$ is a geometric $\shapeset$-pattern locally allowed by $\locallyallowed$, such that there is a \emph{center} $x\in \support{P}$ such that the (square) ball $B(x,r)$ is included in $\support{P}$ and $P$ is minimal for $x$ in the sense that removing a tile from $P$ breaks simple connectedness or the condition on $\support{P}$.

By analogy, a \emph{labelled $r$-pattern} is a labelled $\tileset$-pattern whose projection is a $r$-pattern.
\end{definition}

Notice that the minimality condition is relative to the choice of center $x$, so a $r$-pattern needs not be minimal for inclusion, in the sense that a $r$-pattern may be strictly included in another $r$-pattern. However, different choices of $x$ may be possible for the same $r$-pattern.

Let $D_{max}$ be the maximal diameter of the shapes in $\shapeset$. Given a $r$-pattern $c$, we may assume that its center is $0$ up to translating the pattern. Since a tile in $c$ either intersects $B(0,r)$ or fills a hole left by tiles that intersect $B(0,r)$, it follows that $c$ is contained in $B(0, r+D_{max})$. By finite local complexity, the set of labelled $r$-patterns is finite.

The forbidden patterns defined below are all $D_{max}$-patterns.

\subsubsection{Grid structure}
\label{sec:grid_structure}

\begin{definition}[Grid-like patterns]\label{def:gridlike}
A labelled (or placeholder) pattern $P$ is \emph{grid-like} if there is a grid $g_{x_0, y_0} \coloneq \{(x,y)\in \RR^2 | (x-x_0) \in \ZZ \vee (y-y_0)\in\ZZ \}$ such that:
\begin{enumerate}
\item the intersection of $g_{x_0, y_0}$ with the tile boundaries and the placeholder grid boundaries does not contain any segment;
\item $g_{x_0, y_0}$ contain no tile triple point, that is, a position that belong to the boundary of three tiles or more;
\item if a horizontal (resp. vertical) line in $g_{x_0, y_0}$ crosses the interior of a cell $c$, then $c$ is a row cell (resp. column cell).
\end{enumerate}
\end{definition}
This means that a cell whose interior intersects both a horizontal and vertical line of the grid $g_{x_0,y_0}$ is a coding cell. If the coding cell is a full cell (not cut by the boundary of a tile), it contains a grid vertex.

We define $G_{grid}$ as the set of all non grid-like $D_{max}$-patterns, which is a finite set.

Recall that by hypothesis, $D_{max}$ is not an integer and $D_{max}\geq \sqrt{2}$ as the interior of any tile contains a unit square.
Note that, given a labelled pointed $D_{max}$-pattern $P$ and a choice of center $x$, $P$ contains the tile containing $x$ and all of its neighbours.

Grid-like labelled patterns exist on any geometric pattern.
Indeed, take a pattern $P_0$ of unlabelled shapes and the induced pattern $P_1$ of checkered tiles.
Given a cell $c$ of a checkered tile, the set of positions $(x_0,y_0)$ in the interior of $c$ such that the grid $g_{x_0,y_0}$ contains a tile triple point or a segment intersection with the tile boundaries or the placeholder grid boundary is countable (by Proposition~\ref{prop:countable-segments}, shapes have at most countably many maximal horizontal and vertical segments). By a simple counting argument, there are uncountably many positions $(x_0,y_0)$ in the interior of $c$ whose grid satisfies conditions 1 and 2, and for each such position it is straightforward to check that there is a unique choice of placeholder tile for each checkered tile that satisfies condition 3: since each tile contains a vertex of the grid, this forces the position of a coding cell. 

\begin{figure}[htp]
  \center \includegraphics[width=0.5\textwidth]{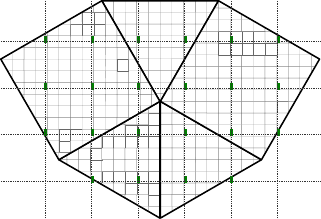}
  \caption{Finding a valid grid placement in a checkered pattern. In green: the set of valid choices for the center of the grid.}
  \label{fig:grid-like_position}
\end{figure}

In a coloured pattern, the set of valid choices for $(x_0, y_0)$ (modulo $1$) for condition 3 is contained in a rectangle drawn in green in Fig.~\ref{fig:grid-like_position}. This stems from the fact that each placeholder tile contains a full coding cell which yield a square of possible grid positions of side length $1/4$, so the set of valid choices is a finite intersection of squares. The pattern is grid-like iff this rectangle is not empty, as conditions 1 and 2 remove only countably many positions, and so deciding if a pattern is grid-like is a computable process.

\subsubsection{Coding forbidden patterns}

\begin{definition}[Coherent patterns]
A grid-like pattern is \emph{coherent} if two coding cells at distance at most $\tfrac{1}{4}$ have the same colour.
\end{definition}

The set of forbidden patterns $G_{coding}$ contains every (coloured) grid-like $\patternsize$-patterns that are not coherent. This ensures a global coding consistency in allowed patterns and tilings.

\begin{figure}[htp]
  \center
  \includegraphics[width=\textwidth]{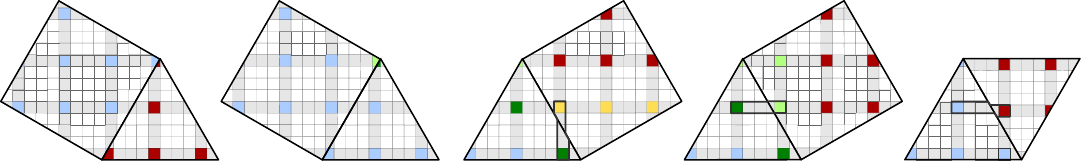}
  \caption{To the left: two non-coherent patterns. Any $\patternsize$-pattern containing one of them is forbidden in $G_{coding}$.\\
  To the right: three patterns that induce a forbidden patterns from $F$. Any $\patternsize$-pattern containing one of them is forbidden in $G_F$.}
  \label{fig:forbidden_GF}
\end{figure}

\subsubsection{Induced forbidden patterns}

To each coherent grid-like pattern $P$, we canonically associate an induced $\mathbb{Z}^2$-pattern $P'$ as follows. Pick some valid grid placement $g_{x_0,y_0}$. If $n$ and $m$ are such that $(x_0+n,y_0+m)\in \support{P}$, define $P'_{n,m}$ as the colour of the coding cell containing $(x_0+n,y_0+m)$ (by the grid-like hypothesis). All valid placements differ by at most $1/4$, so by coding coherence, every valid placement yields the same symbolic pattern $P'$.

We now define $G_F$ as the set of coherent grid-like $\patternsize$-patterns $P$ whose induced $\mathbb{Z}^2$ pattern $P'$ contains a forbidden pattern from $F$. Notice that, for a $\patternsize$-pattern $P$, the support of $P'$ contains a $2 \times 2$ square since $\patternsize\geq \sqrt{2}$.


\subsection{Equivalence with the original domino problem}

We prove Theorem~\ref{thm:main} by reduction to the classical nearest-neighbour domino problem. 

Given a geometric tileset $\tileset$ and an instance $(A,F)$ of the nearest-neighbour $\ZZ^2$-domino problem, define $X'$ as the set of $\tileset$-tilings $\tiling$ that avoid the forbidden patterns $G = G_{grid} \cup G_{coding} \cup G_F$ defined above and such that $\pi(\tiling)\in X$, that is, removing the colours gives a tiling in $X$.

We prove that there is a valid configuration $c$ for $(A,F)$ if, and only if, $X'$ is not empty. 

\begin{proof} [($\Rightarrow$)]
Let $\tiling$ be a geometric tiling in $X$ (which is nonempty by assumption), and let $\tiling_0$ be the corresponding checkered tiling (see Section~\ref{sec:placeholder}).
  
For $x_0, y_0\in\RR^2$, define the unit grid $g_{x_0,y_0} = \{(x,y)\ :\ (x-x_0)\in\ZZ \vee (y-y_0)\in\ZZ \}$. 
 
There are countably many triple points and horizontal or vertical segments in the boundaries of the tiles and of the cells, so there is a choice of $x_0$ and $y_0$ such that $g_{x_0,y_0}$ intersects none of them (these are the conditions for being grid-like as in Definition~\ref{def:gridlike}).

Now define the placeholder tiling $\tiling_1$ as follows: for every tile of $\tiling_0$, we choose the only placeholder tile that corresponds to $g_{x_0,y_0}$ in the sense that a cell that intersects a row, resp. a column, of $g_{x_0,y_0}$ is a row cell, resp. a column cell (it is a coding cell if it intersects both). 

By construction, we did not create a pattern from $G_{grid}$, since every $\patternsize$-pattern is grid-like using the grid $g_{x_0,y_0}$.

We define a coloured tiling $\tiling'$ as follows: for each coding cell, if $(x_0+n, y_0+m)$ is the closest vertex of the grid $g_{x_0,y_0}$, then the cell gets assigned the colour $c_{n,m}$.

\begin{figure}[htp]
    \center
    \includegraphics[width=0.8\textwidth]{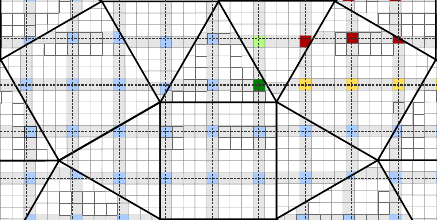}
    \caption{Construction of the coloured tiling $\tiling'$ from a geometric tiling $\tiling$ and a configuration $c$.}
  \end{figure}
  
Since each tile contains at least one vertex of $g_{x_0,y_0}$, the cell containing this vertex is a coding cell, so every codng cell in this tile is at distance at most $1/4$ from a vertex of $g_{x_0,y_0}$. It follows that two different coding cells that are close get assigned the same vertex and no pattern from $G_{coding}$ may appear.

By construction, $\tiling'$ contains no pattern in $G_F$ since $c$ contains no pattern in $F$.
  
We have proved that $\tiling'$ does not contain any pattern from $G$. Since $\pi(\tiling')=\tiling \in X$, $\tiling'\in X'$.
\end{proof}

Notice that the construction of $\tiling'$ from $\tiling$ may not be computable (especially relative to the choice of $x_0$ and $y_0$) but this is not a part of reduction.

\begin{proof}[($\Leftarrow$)]
Let $T'$ be a tiling in $X'$. We build a function $\position : \mathbb{Z}^2 \to \mathbb{R}^2$ such that for any $(n,m)$, $\position(n,m)$ is in the interior of a coding cell. This defines a configuration $c$ where $c_{n,m}$ has the colour of the cell containing $\position(n,m)$.

Up to translating the whole tiling, assume that $(0,0)$ is in the interior of a coding cell of $T'$. As the $\patternsize$-pattern $P$ around $(0,0)$ is grid-like, choose a position $\position(0,0)$ at distance at most $\tfrac{1}{4}$ from $(0,0)$ such that $g_{\position(0,0)}$ is a valid grid position for $P$.

We repeat this process inductively around positions already defined, by increasing distance from $(0,0)$. Let $(n,m)\in\mathbb{Z}^2$ such that $\position(n,m)$ is defined but not $\position(n+1,m)$ (other directions are similar).
As the pattern $P_{n,m}$ around $\position(n,m)$ is grid-like with $g_{\position(n,m)}$ as a valid grid placement, the position $\position(n,m) + (1,0)$ is vertex of $g_{\position(n,m)}$ so it is inside a coding cell.
As the pattern $P_{n+1,m}$ around $\position(n,m)+(1,0)$ is grid-like, pick a position $\position(n+1,m)$ at distance at most $\tfrac{1}{4}$ from $\position(n,m)+(1,0)$ such that $g_{\position(n+1,m)}$ is a valid grid position of $P_{n+1,m}$.\bigskip

Even though this process is not deterministic, it is possible to prove that the configuration $c$ does not depend on the order in which the values of $\position$ are defined or on the choices made in the construction. For simplicity, we prove the weaker statement that any forbidden pattern in $c$ yields a pattern in $G_F$.

Suppose that the pattern $c_{(0,1), (1,1)}$ is forbidden; the same proof holds for other positions and directions. If $\position(1,1)$ was defined from $\position(0,1)$, this would create a pattern from $G_F$ by construction. Assume that $\position(1,1)$ was defined from $\position(1,0)$. Since $\position(0,1)$ and $\position(1,0)$ were both defined from $\position(0,0)$, it follows that $p = \position(1,0) + (0,1)$ is at distance at most $\tfrac{1}{2}$ from $p' = \position(0,1)+(1,0)$ (by triangular inequality), and they are both inside some coding cell (in possibly different tiles $t$ and $t'$).

Consider the $\patternsize$-pattern $P$ around $\position(1,1)$: it contains $p'$ and so it contains the whole tile $t'$.
By assumption on the size of the tiles, $t'$ contains some vertex of the grid centered on $\position(1,1)$ which is therefore in a coding cell. $p'$ is in a coding cell of $t'$ so, by definition of the placeholder tile, it is at distance at most $\tfrac{1}{4}$ from a vertex $v$ of the grid. It follows that:
\begin{align*}d(\position(1, 1), v) &\leq d(\position(1, 1), \position(0,1)+(1,0))\\&\qquad + d(\position(0,1)+(1,0), \position(1,0)+(0,1)) + d(\position(1,0)+(0,1),v) \\&< \tfrac{1}{4} + \tfrac{1}{2} + \tfrac{1}{4} = 1,\end{align*}
which implies that $\position(1, 1) = v$ (they are vertices on a unit grid). Since $p'$ is in a coding cell and is at distance $<\tfrac{1}{4}$ from a coding cell with colour $c_{1,1}$, the two cells have the same colour by $G_{coding}$. It follows that the grid centered on $\position(0,1)$ "sees" the pattern $c_{(0,1), (1,1)}$, so the $\patternsize$-pattern around $\position(0,1)$ belongs to $G_F$.
\bigskip

We now prove that $c$ is valid for $F$.
The previous argument shows that any pattern $P'\sqsubset c$ of size $2\times 2$ appears in a grid-like $\patternsize$-pattern. The tiling $T'$ avoids $G_F$, so $P'$ does not contain any forbidden pattern in $F$.
\end{proof}

\begin{remark}
Since tilings in $X$ can be arbitrarily complex, there is no reason for $\domino_X$ to be $\Pi_1$-complete. However, if there is an algorithm that, given $r$, enumerates all central $r$-patterns of all $t\in X$, we have a semi-algorithm for $\domino_X$ by trying all possible colourings of these patters and rejecting if we found no valid colouring. When $X$ is a subshift, it is enough to have such an algorithm for enumerating locally allowed patterns by a compactness argument; this corresponds to effective subshifts. If this case $\domino_X$ is $\Pi_1$-complete.
\end{remark}

\subsection{Hidden quasi-isometries}\label{sec:quasi-isometry}

We notice that our reduction to the domino problem relies on the existence of a quasi-isometry between the adjacency graph of the tiling and $\mathbb Z^2$; since quasi-isometries have appeared in other contexts as an invariant between spaces that preserve the difficulty of the domino problem, we discuss this phenomenon in this section.

\begin{definition}[Adjacency graphs]Fix a set of geometric tiles, a geometric tiling $x$ with FLC and an enumeration of all possible $2$-tile patterns in $x$. The adjacency graph $G(x)$ is the graph whose vertices correspond to tiles, edges link tiles that are adjacent in $x$, and each edge is labelled according to the corresponding $2$-tile pattern. 
\end{definition}

The problem $\domino_{\{x\}}$ can be seen as the domino problem on the graph $G(x)$. Most existing results regarding the domino problem take place on a single fixed Cayley graph. A specificity of the $\domino_X$ problem is that it corresponds to a domino problem on a family of labelled graphs $G(X) = \{G(x) : x\in X\}$\footnote{This family can be described by forbidding subgraphs that correspond to forbidden patterns in $X$, which means that it is a graph subshift in the sense of \cite{arrighi2023}.}. Still, our result has some similarities with existing results on fixed graphs.

\begin{definition}[Quasi-isometry]A map $\phi : X\to Y$ between metric spaces $(X,d_X)$ and $(Y,d_Y)$ is a \emph{quasi-isometry} if there exists $k>0$ such that
\[
  \frac{1}{k}\cdot d_X(x,x') - k \leq d_Y(\phi(x),\phi(x')) \leq k\cdot d_X(x,x') + k 
\quad \text{ and } \quad
 Y = \bigcup_{x\in X} B(\phi(x), k).
\]
\end{definition}

\begin{remark}
  The adjacency graph of a geometric tiling with finite local complexity is planar and quasi-isometric to $\mathbb{Z}^2$. 
\end{remark}

In fact, our reduction yields an explicit quasi-isometry $G(x)\to \mathbb Z^2$ for every $x\in X$: draw the grid on $x$, fix an origin $(0,0)$ for the grid arbitrarily, and define $\phi : G(x)\to \mathbb Z^2$ such that $\phi(v)$ is any element of $\mathbb{Z}^2$ in the grid inside $v$.

In the case of finitely presented groups it is known that the undecidability of the domino problem is a quasi-isometry invariant \cite{cohen2015}, in addition to other properties such as the existence of strongly aperiodic subshifts \cite{barbieri2022}. In particular, the domino problem on any Cayley graph that is quasi-isometric to $\mathbb{Z}^2$ is undecidable. However, adjacency graphs for geometric tilings are very far from Cayley graphs; for example, they are not (vertex-)transitive. 

Shortly after the present paper was submitted, Barbieri and Bitar \cite{barbieri2025} extended these invariance results under quasi-isometry beyond Cayley graphs to structures called blueprints; see in particular Section 6 in \emph{op.cit} that relate to geometric tilings. While their results are more general than Theorem~\ref{thm:main} on most aspects (they apply to all groups and not only $\mathbb Z^2$, for example), they only cover the case where the set of geometric tilings is a subshift as of the currently available version\footnote{Following personal communication with the authors, a future version of \emph{op. cit.} may apply to arbitrary sets of tilings.}.

\section{Finite local complexity is undecidable for geometric tilings}\label{sec:undecFLC}

Given the important role played by the FLC hypothesis in carrying symbolic results to the geometric setting, we study the decidability of this property.

\begin{definition}
The FLC problem for purely geometric subshifts:
  \begin{description}
  \item[Input:] a shapeset $\shapeset$ of polygonal tiles.
  \item[Output:] does the purely geometric subshift on $\shapeset$ have finite local complexity?
  \end{description}
\end{definition}

\begin{theorem}
The FLC problem for purely geometric subshifts is $\Sigma_1^0$-complete and therefore undecidable.
\end{theorem}

We start by proving the upper bound, namely that it is recursively enumerable.

\begin{lemma}
	The FLC problem is in $\Sigma_1^0$.
\end{lemma}
\begin{proof}
	When a geometric subshift does not have FLC, it has tilings with a shear line (see \cite{kenyon1992}): that is to say that two independent half planes can be tiled independently and shifted continuously relative to each other along this line. In order to prove that FLC is a $\Sigma^0_1$ property, we will show that having a shear line is a $\Pi_1^0$ property.

	When a shear line appears, then its direction is necessarily the direction of one of the sides of some shape of the input. There is a finite number of such directions. A semi-algorithm to decide the existence of a shear line thus consists in finding if one of these directions cuts the plane in two halves such that each half plane may be perfectly tiled.

	As the existence of a tiling for a shapeset is a $\Pi^0_1$ problem, so is the existence of a tiling of the half plane. The search for tileable half planes along the finite directions of the sides of the shapes may be done concurrently. If no direction can be tiled in perfect half planes, then at some points the finite number of searches will all halt and the algorithm halts, otherwise it runs forever.
\end{proof}

\begin{lemma}
	The FLC problem is $\Sigma_ 1^0$-hard.
\end{lemma}
\begin{proof}
We proceed by reduction to the halting problem for Turing machines. For each Turing machine $M$, we build a shapeset $\shapeset$ such that the purely geometric subshift is not empty and has finite local complexity iff $M$  halts on the empty tape.\bigskip

We start from $1\times 1$ square tiles decorated with the Robinson~\cite{robinson1971} decorations as modified by Gurevich and Koryakov~\cite{gurevich1972}: this tileset tiles the plane periodically iff the Turing machine encoded in it halts and aperiodically otherwise, as shown in \autoref{fig:robGK}.
As we want to effectively embed computations in non-intersecting hierarchical squares of Robinson tilings, we use a 10-tiles variant of the Robinson tileset where odd levels of the hierarchy of squares are blue and even levels of the hierarchy of squares are red. 
Robinson's tileset can be defined either by purely geometric or symbolic tilings: see Fig.~\ref{fig:robinson_alternating_squares}. Gurevich and Koryakov's tileset can also be made purely geometric: see Figs.~\ref{fig:gk_tiles} and \ref{fig:robinson-gk_computation}.

In the periodic case, the only tilings consist of the smallest Robinson squares containing the whole computation repeated periodically, as shown in Fig.~\ref{fig:GK_period}. The aperiodic case consists either of only Robinson structures, or of two half-planes or four quarter-planes of Robinson structures separated by a shear line of Gurevich-Koriakov border, illustrated on Fig.~\ref{fig:GK_infinite_sliding}.
\pgfdeclareimage{robinson}{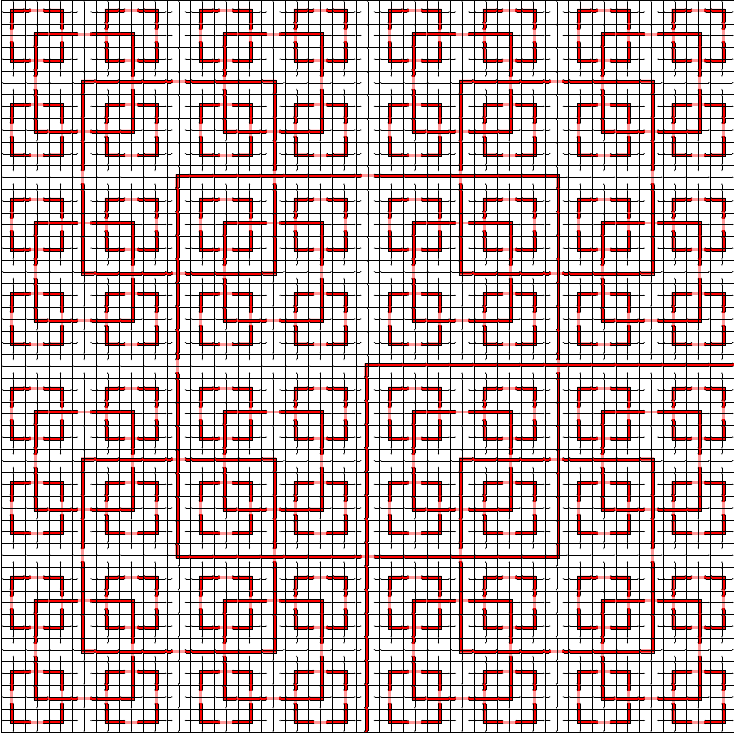}
\definecolor{greenborder}{HTML}{2d9e2f}

\tikzset{robinsonRect/.pic={
\begin{scope}
 \clip (3,3) rectangle +(6.8,6.8);
 \node[anchor=south west] at (0.05,0.08) {\pgfuseimage{robinson}};
 \fill[color=greenborder] (3.2,3) rectangle +(6.4,.2);
 \fill[color=greenborder] (3.2,9.8) rectangle +(6.4,-.2);
 \fill[color=greenborder] (3,3.2) rectangle +(.2,6.4);
 \fill[color=greenborder] (9.8,3.2) rectangle +(-.2,6.4);
 \draw (9.8,3) rectangle +(6.8,6.8);
\end{scope}
}}

\begin{figure}

\scalebox{.5}{
 \begin{tikzpicture}
  \begin{scope}[shift={(15,-.5)}]
  \pic at (0,0) {robinsonRect};
  \pic at (6.8,0) {robinsonRect};
  \pic at (6.8,6.8) {robinsonRect};
  \pic at (0,6.8) {robinsonRect};
  \end{scope}

  \begin{scope}
    \node[anchor=south west] at (3.15,3.125) {\pgfuseimage{robinson}};
  \end{scope}

 \end{tikzpicture}
 }

 \caption{\label{fig:robGK} Two tilings by the Gurevich-Koryakov tileset. On the left: when the embedded Turing machine does not halt, it gives aperiodic tilings. On the right: when it does halt, the smallest square that witnesses the halting state gets a special colour on the border that forces a periodic tiling.}
\end{figure}
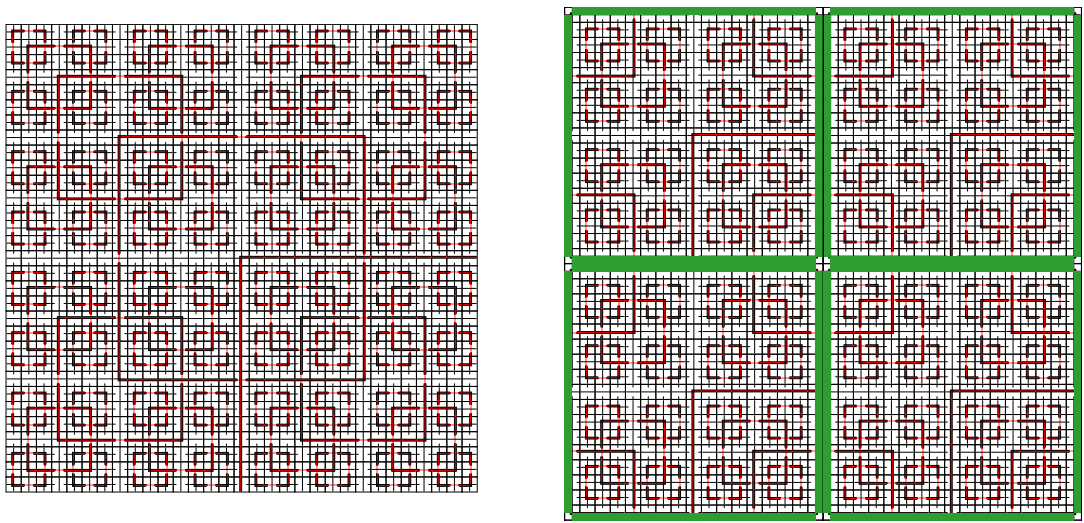


\begin{figure}[htp]
  \center
  \includegraphics[height=4cm]{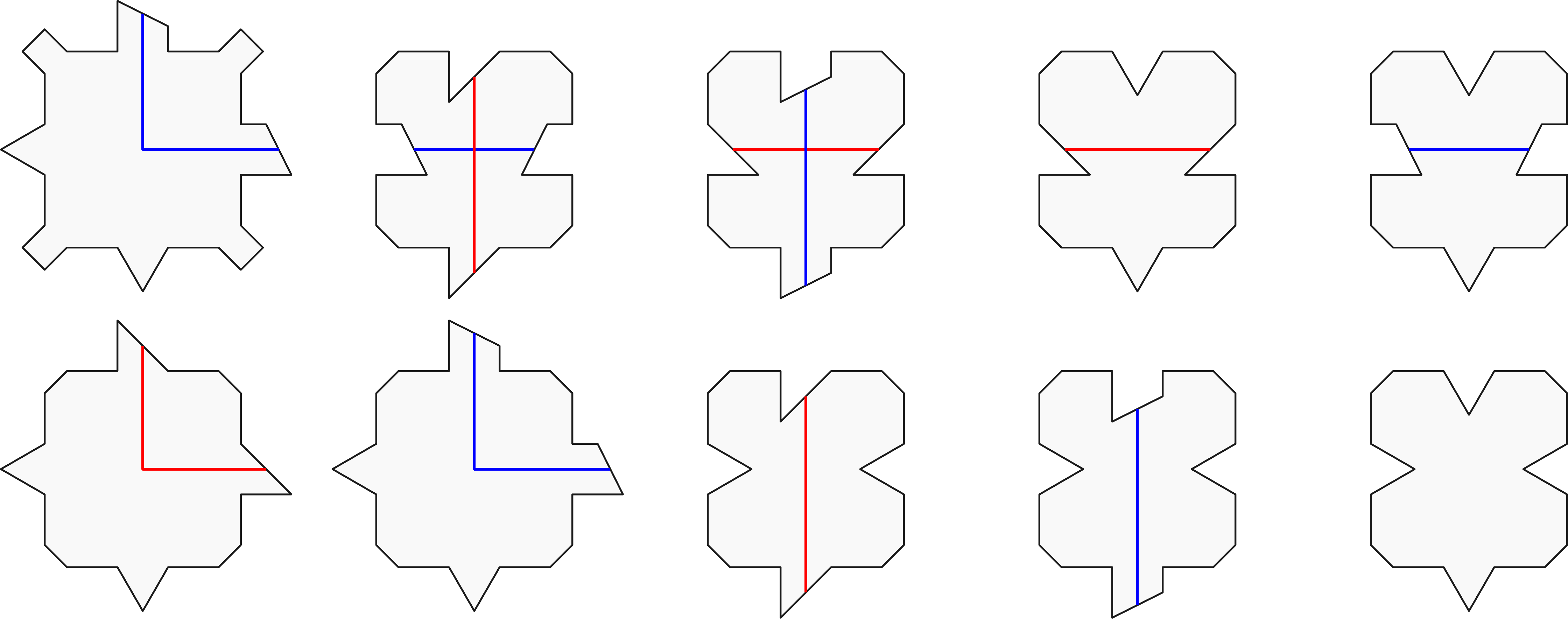}
  \caption{A purely geometric presentation of the alternating squares variant of Robinson tiles.
    Decorations are shown for clarity only, and the computing layer is not represented.}
  \label{fig:robinson_alternating_squares}
\end{figure}

\begin{figure}[htp]
  \center
  \includegraphics[height=4cm]{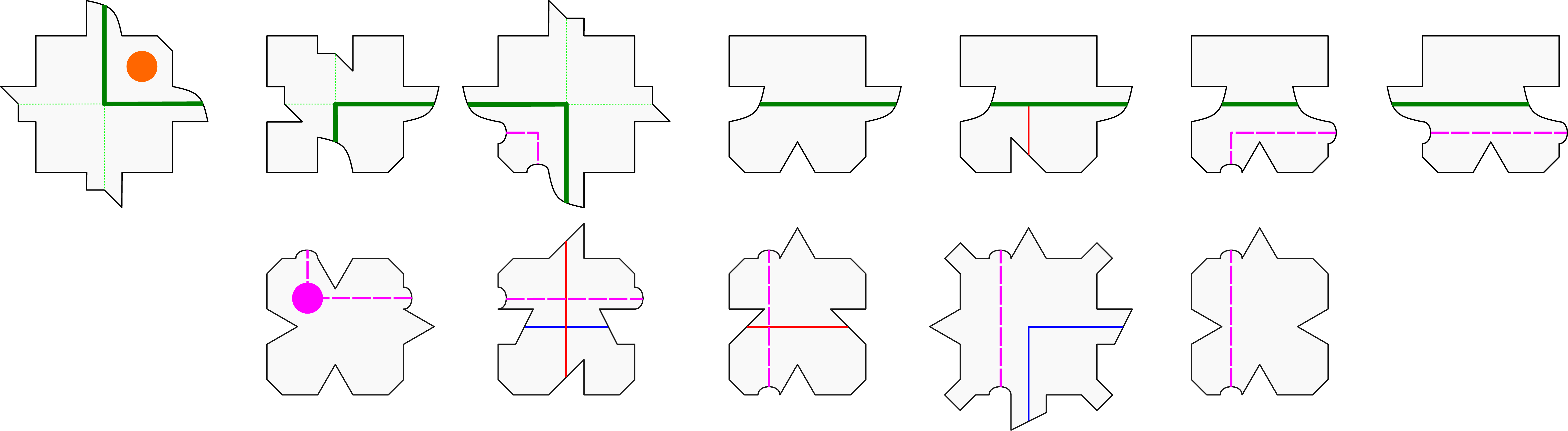}
  \caption{New tiles (in addition to Fig.~\ref{fig:robinson_alternating_squares}) giving the Gurevich-Koryakov variant.
  If a tile with the violet signal appears (bottom row), then the next square in the hierarchy will be green (top row) which forces a periodic tiling by green squares of the same size. Green squares can not appear otherwise.
  Every tile where the halting state appears on the computing layer must start a violet signal (only one such tile is represented here). 
  Decorations are shown for clarity only, and the computing layer is not represented.
  }
  \label{fig:gk_tiles}
\end{figure}

\begin{figure}[htp]
 
  \begin{subfigure}{0.45\textwidth}
    \center
    \includegraphics[height=4cm]{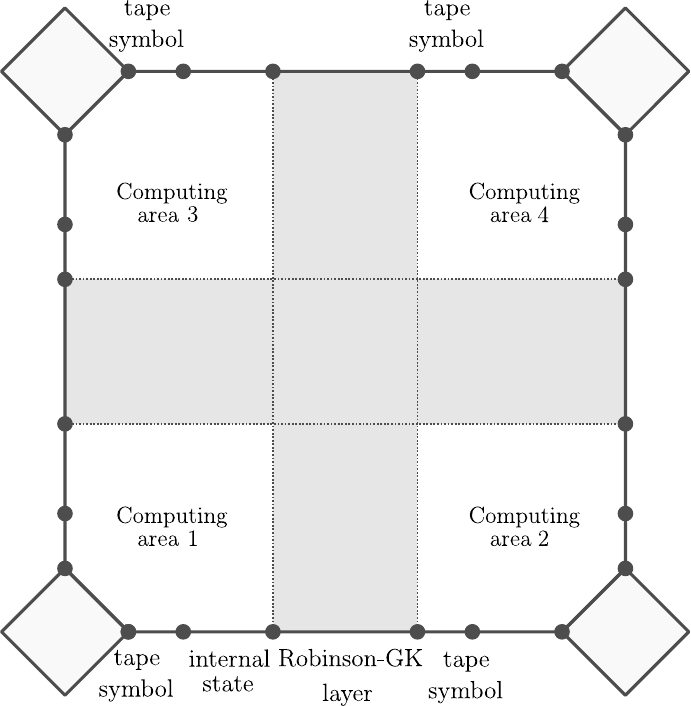}
    \caption{Each Robinson or Gurevich-Koryakov computing tile contains 4 computing areas, each area behaving like a Wang tile of the usual Turing machine simulation tileset (see for example \cite{Jeandel2020}).}
  \end{subfigure}
  \hfill
  \begin{subfigure}{0.45\textwidth}
     \center
    \includegraphics[height=4cm]{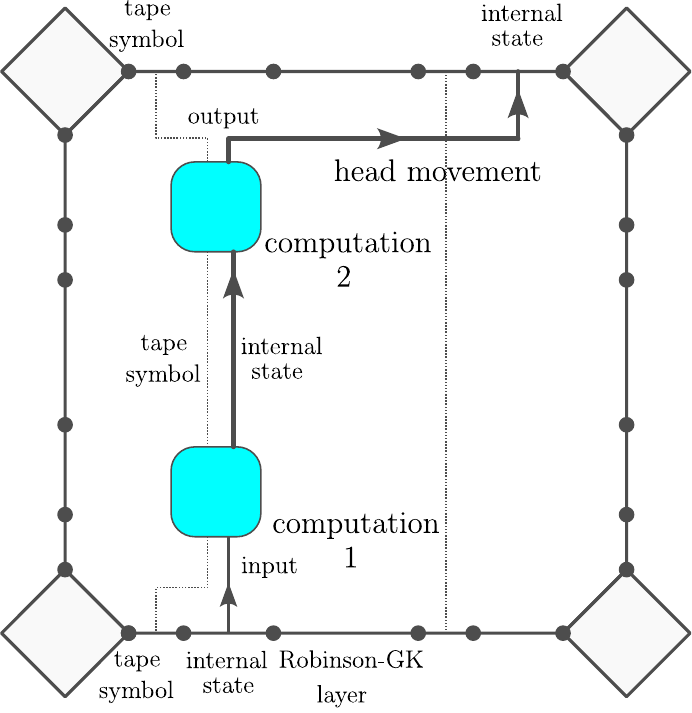}
    \caption{Schematics of a computing tile where the Robinson or Gurevich-Koryakov layer is abstracted. Each tile simulates two computation steps of a Turing machine.}
  \end{subfigure}

  \begin{subfigure}{\textwidth}
    \center \includegraphics[height=3cm]{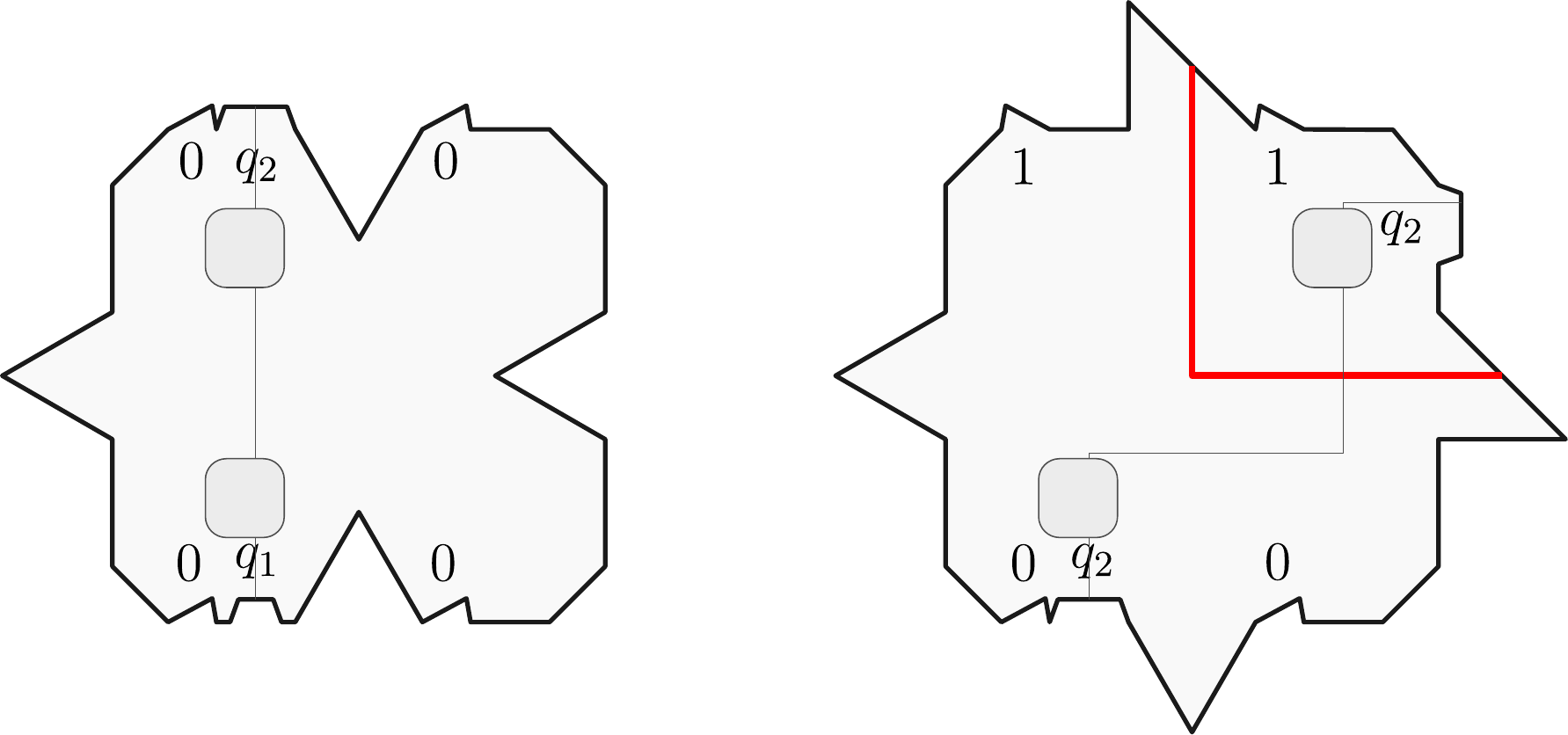}
    \caption{Two examples of computing tiles. The computing path is represented with thin gray lines and boxes.\smallskip\\
      Left: starting from state $q_1$, the machine reads a $0$ and performs two computation steps that result in writing a $0$ and going in state $q_2$ without moving (the intermediate state and symbol is not specified).\\
      Right: starting from state $q_2$, the machine reads a $0$, writes a $1$, stays in state $q_2$ and moves to the right. Then, the machine again reads a $0$, writes a $1$, stays in state $q_2$ and moves to the right.}
    \end{subfigure}
  \caption{Embedding computation in tiles.
    Additional \emph{auxiliary tiles} are used to initialise the tape and connect the computation areas.} 
  \label{fig:robinson-gk_computation}
\end{figure}

\begin{figure}[htp]
  \center
  \includegraphics[height=3cm]{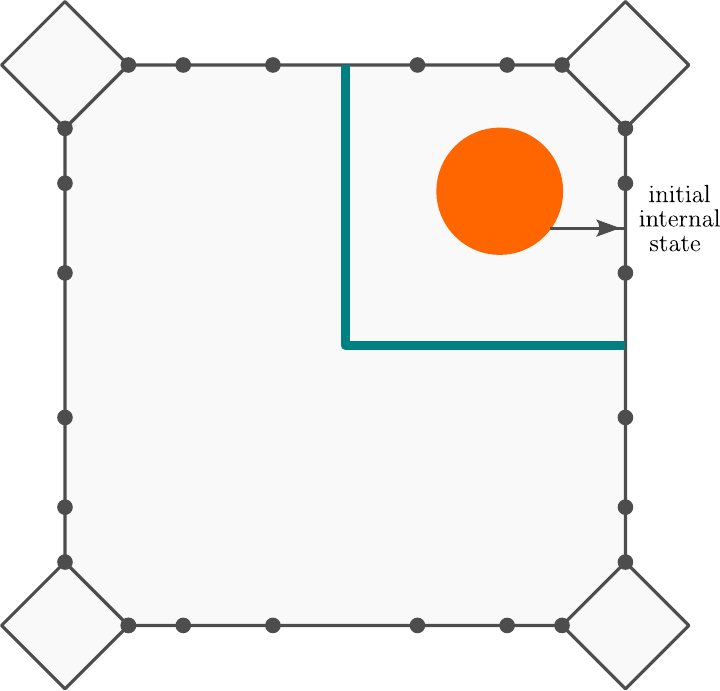}\hspace*{3cm}
  \includegraphics[height=3cm]{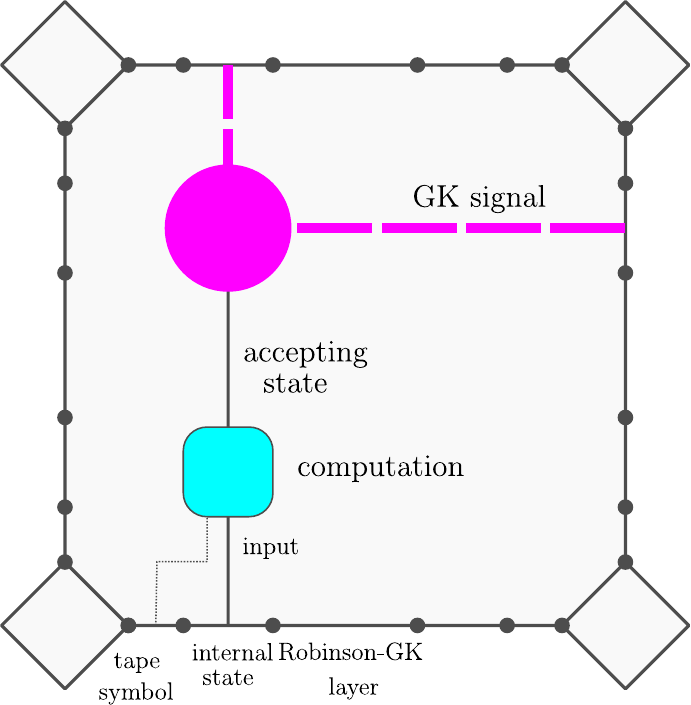}
  \caption{Right: the seed tile starts the computation (the square can be blue or green). Left: any computation tile that reaches the halting state triggers a violet signal.}
  \label{fig:robinson-gk_seed-and-signal}
\end{figure}

\begin{figure}[p]
  \center
  \includegraphics[width=\textwidth]{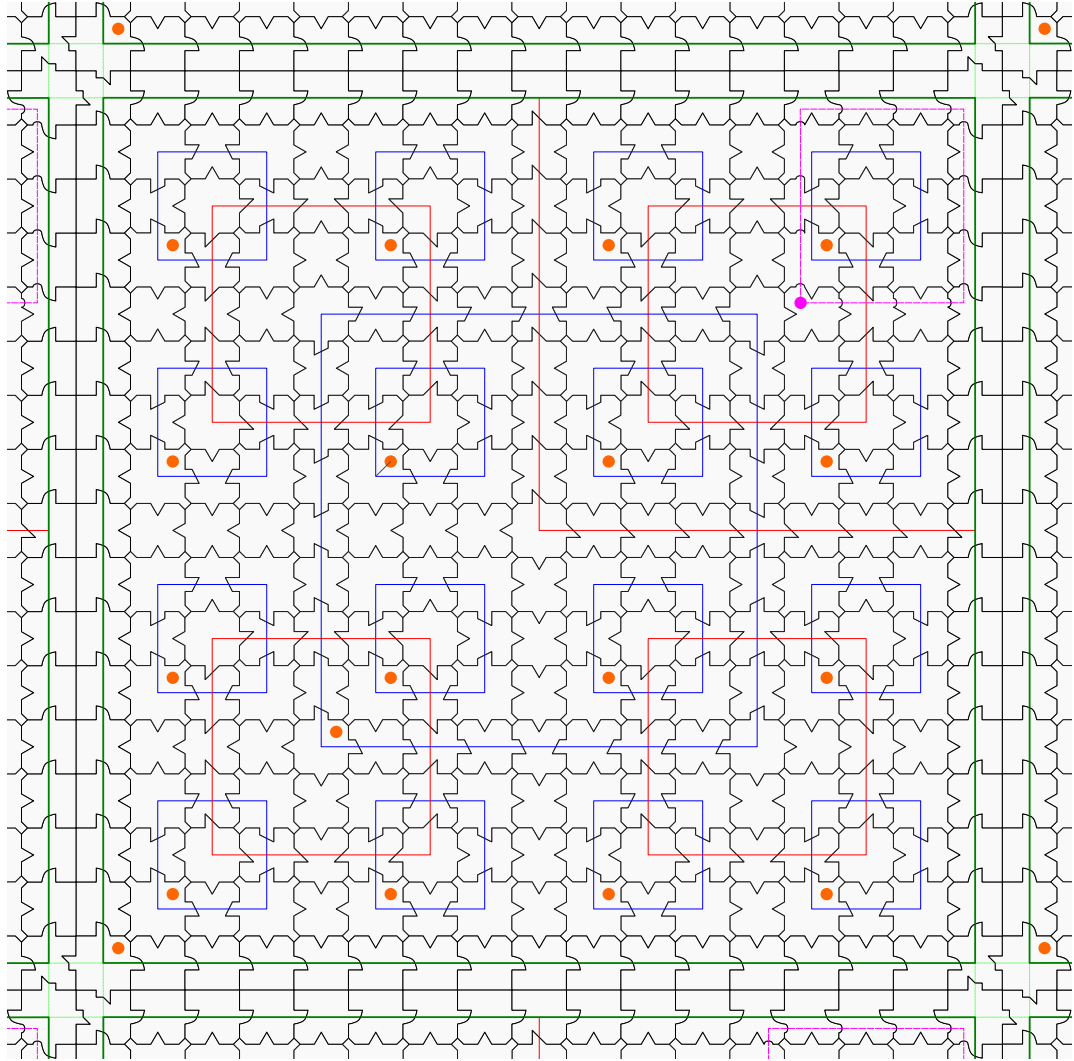}
  \caption{A period of the Gurevich-Koryakov tiling. The computation layer is not represented.}
  \label{fig:GK_period}
\end{figure}

\begin{figure}[p]
  
  \includegraphics[width=\textwidth]{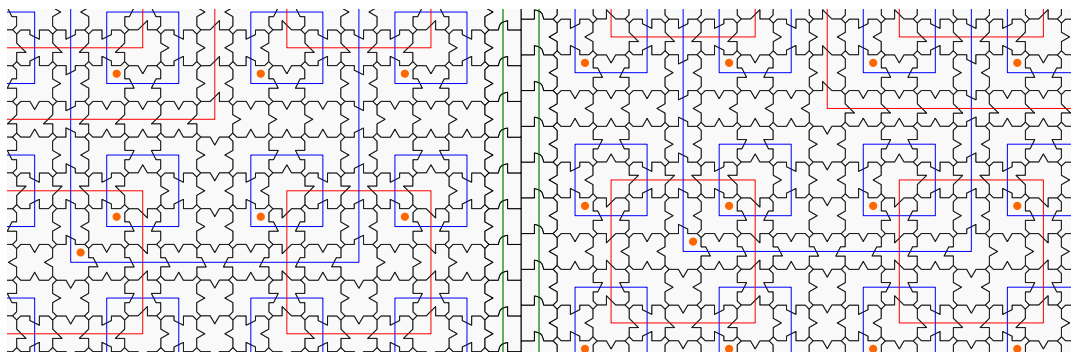}
  \caption{If the simulated machine does not halt, some tilings consist of two "infinite green squares", that is,  independent half-planes. Such half-planes may be shifted continuously relative to each other as only the corners of the green squares synchronise.}
  \label{fig:GK_infinite_sliding}
\end{figure}

Let us adapt this construction to our setting. For a given Turing machine, we define the associated geometric subshift using Gurevich and Koryakov's tileset with one modification: the outer borders of the green squares do not have any markings/indentations except at the corners (to enforce periodicity when the machine halts and the green squares are finite; see Fig.~\ref{fig:GK_period}). In the case where the machine does not halt, there exist tilings consisting of two half-planes of Robinson structures separated by a green border line called a \emph{shear line}. Any vertical shift of a half-plane relative to the other is possible as shown on Fig.~\ref{fig:GK_infinite_sliding}. This means that the subshift $X$ does not have finite local complexity when the machine does not halt, so $X$ has finite local complexity if and only if the machine halts.
. 
\end{proof}

\begin{remark}[FLC for subshifts]
  If $M$ does not halt, any individual tiling in the subshift $X$ has finite local complexity (even though $X$ doesn't), since it contains at most one shear line with a given offset which yields only finitely many $2$-shapes patterns. The construction can be tweaked so that individual tilings do not have FLC in this case by duplicating our shapeset into squares of size 1 and squares of size $\varphi = (1+\sqrt{5})/2$ (or any irrational number). Tiles inside a finite or infinite green square synchronise so that only one size is used. However, when there is a shear line (in the case where the machine does not halt), the two half-planes may use different shapesets, as illustrated on Fig.~\ref{fig:infinite_local_complexity}. The offset between adjacent tiles across the shear lines are the multiples of $\varphi$ modulo $1$ (plus an initial offset) which are dense in the interval $[0,1]$. So such a tiling does not have finite local complexity.
  \begin{figure}[htp]
    \center
    \includegraphics[width=\textwidth]{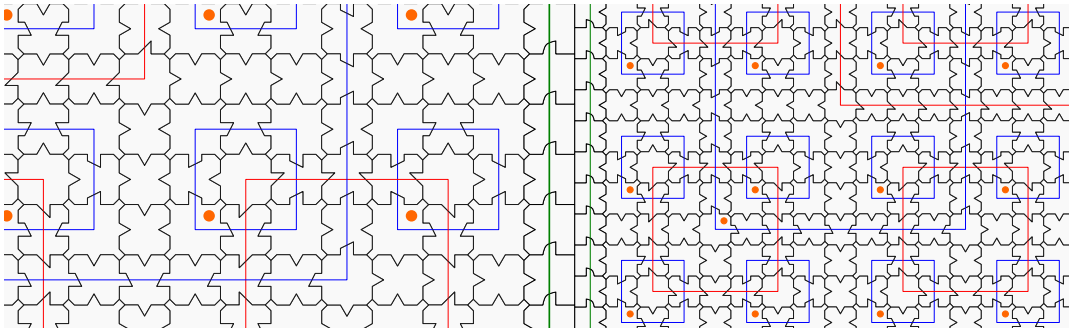}
    \caption{A tiling with infinite local complexity. Tiles in the left half-plane are scaled by $\varphi$.}
    \label{fig:infinite_local_complexity}
    \end{figure}
\end{remark}



\appendix

\section{Aperiodic Wang tilesets from aperiodic geometric tilesets}
In this section, we present a folklore construction to define a Wang tileset from a purely geometric subshift in some cases. This is relatively orthogonal to Section \ref{sec:dominoUndec} where we show how to simulate Wang tilesets in symbolic-geometric tilings on arbitrary shapes.

This construction was first used in \cite[\S 11.1]{grunbaum1987} for Kite-and-Darts Penrose tilings to produce an aperiodic Wang tileset of $32$ tiles, and left as an exercise (11.1.2) for rhombus Penrose tilings to produce an aperiodic Wang tileset of $24$ tiles; full details on this latter case can be found in \cite[\S 5]{jang2021}. We give an overview of this construction, in particular to shed light on geometric and combinatorial conditions that make it work.
We give a somewhat formal statement below. We state the result and present the construction for rhombus tilings, though it can be applied to other shapes e.g. the Kites-and-Darts.

\begin{theorem}[folklore]
   Let $\shapeset$ be a finite set of rhombus shapes and $\tileset$ a finite labelled tileset on $\shapeset$.
   Let $\subshift$ be an edge-to-edge $\tileset$ nearest-neighbour subshift of finite type, that is, defined by finitely many forbidden patterns made of two edge-adjacent tiles.
   If there exists a shape $\shape_0 \in \shapeset$ such that $\shape_0$ is uniformly recurrent in $\subshift$, then there exists a Wang tileset $\tileset_{\subshift,t_0}$ whose subshift of valid tilings is orbit-equivalent to $X$.

   In particular if $X$ is aperiodic, then $\tileset_{\subshift,t_0}$ is an aperiodic Wang tileset.

   Moreover if the language of $X$ (at least up to the radius of recurrence of $t_0$) can be computed, then $\tileset_{\subshift, t_0}$ can be effectively constructed.
\end{theorem}
We do not give a precise definition of orbit equivalence \cite{lind2021}. The idea is that there is a well-behaved correspondence between orbits of tilings in $X$ and orbits of valid $\tileset_{\subshift,t_0}$ tilings, which is illustrated in Fig.~\ref{fig:square-to-wang}.
Note that here we implicitly consider the continuous shift on Wang tilings instead of the usual discrete one.

Let $\shapeset$ be a finite set of rhombus shapes and $\tileset$ a finite labelled tileset on $\shapeset$.
Let $\subshift$ be an edge-to-edge $\tileset$ nearest-neighbour subshift of finite type, that is, defined by finitely many forbidden patterns made of two edge-adjacent tiles.

In edge-to-edge rhombus tilings, the shapes form \emph{chains} or \emph{ribbons} of rhombuses sharing an edge-direction \cite{kenyon1993}, so that each rhombus of edges $\vec{u}$ and $\vec{v}$ is the intersection of two chains of direction $\vec{u}$ and $\vec{v}$ (and conversely). This is illustrated in Fig.~\ref{fig:chains}.

\begin{figure}[htp]
  \center \includegraphics[width=0.5\textwidth]{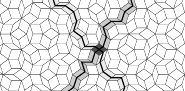}
  \caption{Two chains intersecting in an edge-to-edge rhombus tiling.}
  \label{fig:chains}
\end{figure}

Let $t_0$ be a rhombus with edge directions $\vec{u}_0$ and $\vec{v}_0$.
Call a $t_0$-\emph{chain-square} a pattern delimited by two pairs of chains of same edge direction $\vec{u}_0$ and $\vec{v}_0$, respectively, and intersecting no other chains from these directions, see Fig.~\ref{fig:chain-square}. The four \emph{corners} of such a pattern are four tiles of shape $t_0$.

\begin{figure}[htp]
  \center \includegraphics[width=0.4\textwidth]{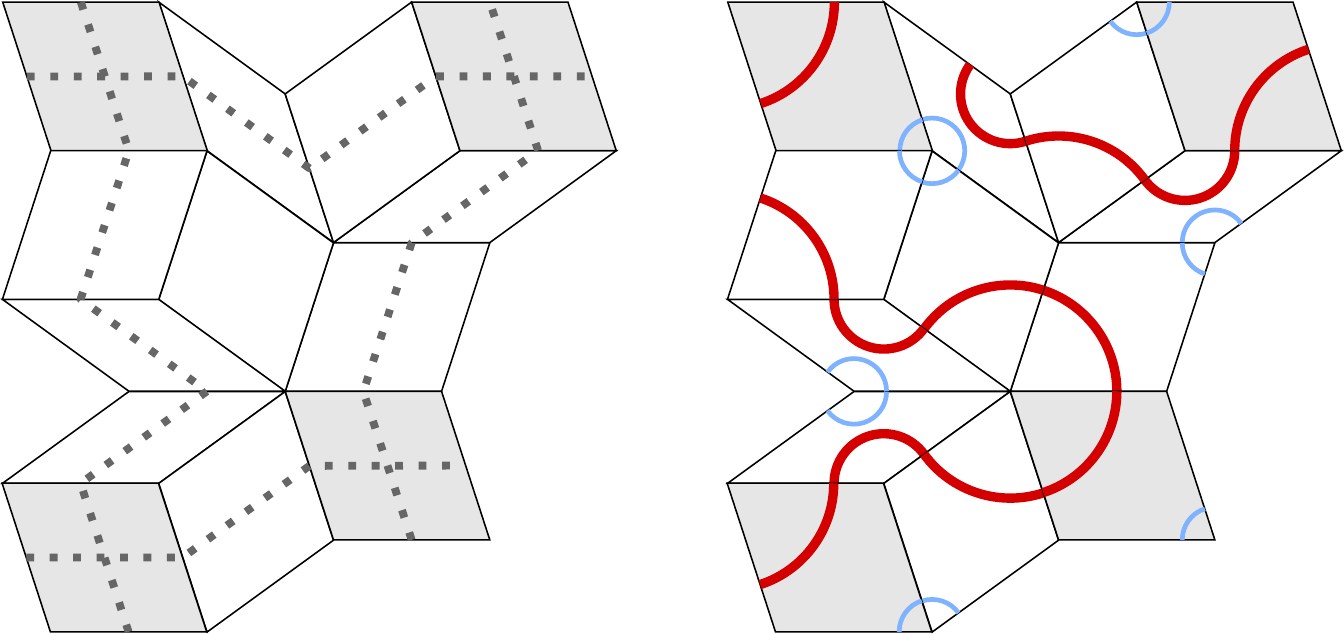}
  \caption{Left: a chain-square. Right: the same shape with the labels.}
  \label{fig:chain-square}
\end{figure}

Assume that $t_0\in \shapeset$ is uniformly recurrent in $\subshift$, that is, there exists $R>0$ such that any disk of radius $R$ in any tiling $\tiling \in \subshift$ contains a tile $t$ such that $\pi(t)=t_0$. Under this assumption, 

\begin{itemize}
\item any tile $t$ in any tiling $\tiling \in \subshift$ belongs to at least one $t_0$-chain-square either as a corner (if $t$ is a copy of $t_0$), in a border (if it shares one edge direction) or in the interior.
\item the size of $t_0$-chain-squares is bounded, so there are finitely many different $t_0$-chain-squares in $\subshift$.
A formal proof of a very similar statement is given in e.g. \cite[Lemma 21]{hellouin2023}.
\end{itemize}

\begin{figure}[htp]
  \center \includegraphics[width=\textwidth]{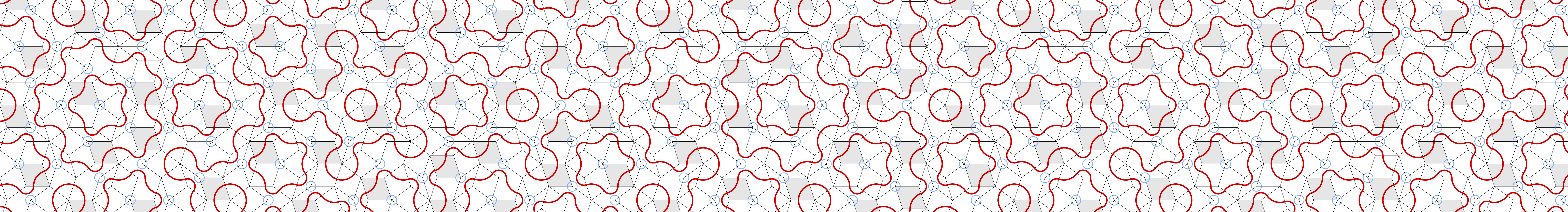}
  \caption{A fragment of labelled Penrose rhombus tiling with occurences of a particular tile highlighted.}
\end{figure}

We define a set $\tileset_{\subshift,t_0}$ of Wang tiles as in Fig.~\ref{fig:square-to-wang}: for each $t_0$-chain-square $P$,  $\mathbf{t}(P)$ is the Wang tile whose labels are the sides of $P$. Horizontal, resp. vertical labels are segments of $\vec{u}_0$-chains, resp. $\vec{v}_0$-chains between two occurrences of $t_0$.

\begin{figure}[htp]
  \center \includegraphics[width=0.5\textwidth]{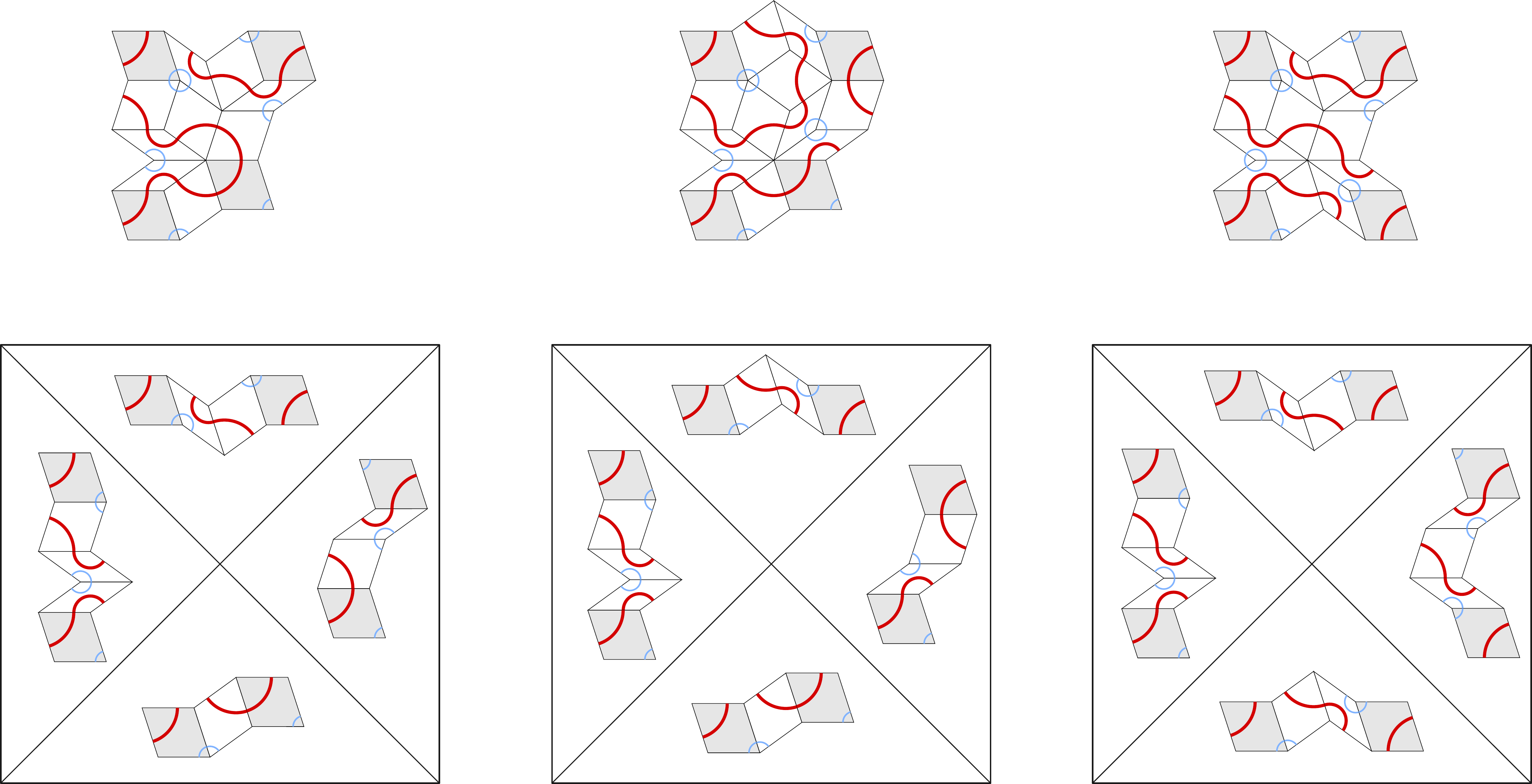}
  \caption{Converting chain-squares to Wang tiles.}
  \label{fig:square-to-wang}
\end{figure}

Each locally allowed pattern of $\tileset_{\subshift,t_0}$ Wang tiles can be converted to a geometric $\tileset$-pattern that is locally allowed in $\subshift$, as is Fig.~\ref{fig:square-patch-wang}, because the subshift $\subshift$ only has nearest-neighbour forbidden patterns. It follows that each tiling in $\subshift$ in which $t_0$ is uniformly recurrent can be converted to a $\tileset_{\subshift,t_0}$ tiling, and conversely.

\begin{figure}[htp]
  \center \includegraphics[width=\textwidth]{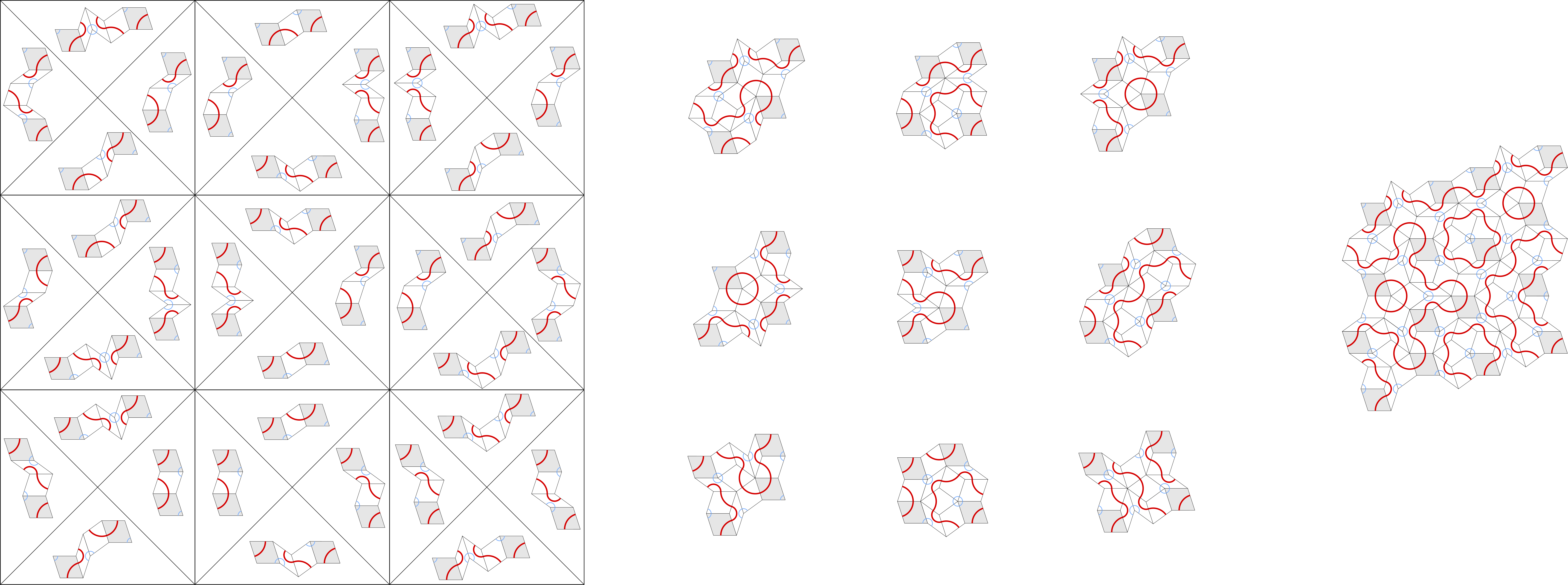}
  \caption{Converting a patch of Wang tiles to a rhombus tiling.}
  \label{fig:square-patch-wang}
\end{figure}

In particular, if $\subshift$ contains no periodic tiling, then neither does $\tileset_{\subshift,t_0}$.
Note however that the symmetries of the geometric tiling are not necessarily preserved. 
Therefore, if $\subshift$ contains no periodic configuration and assuming that one can enumerate the complete set of $t_0$-chain-squares that are locally allowed in $\subshift$, we obtain in this way a corresponding aperiodic Wang tileset.\medskip

\begin{figure}[htp]
  \center \includegraphics[width=\textwidth]{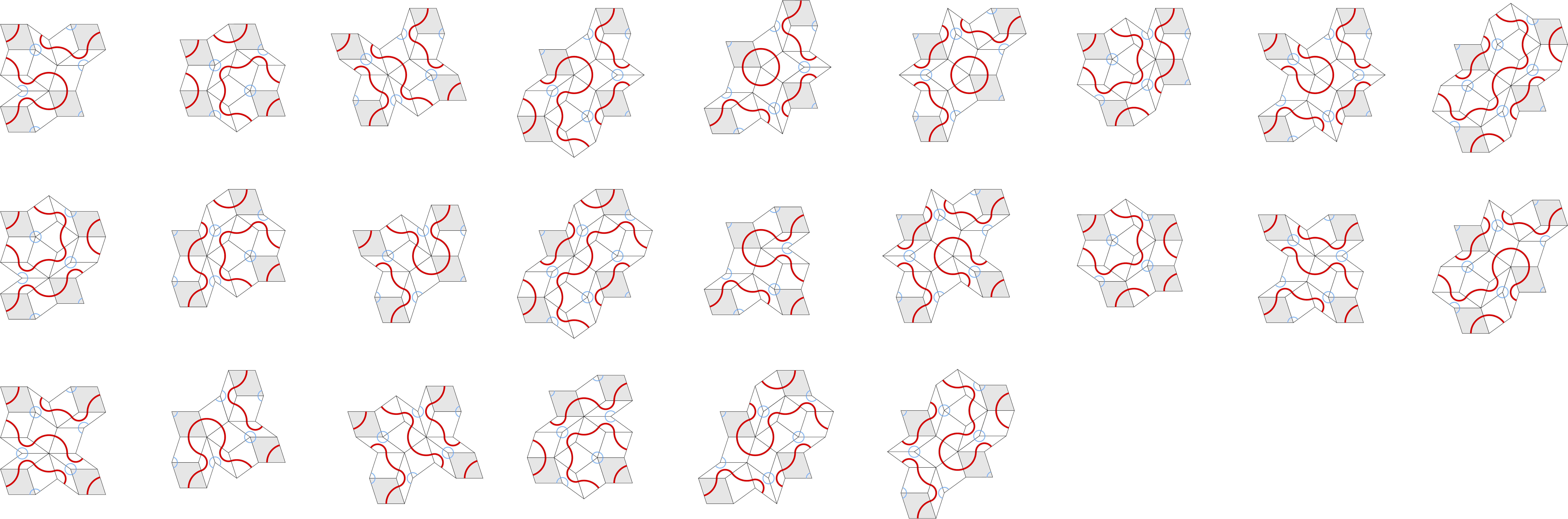}
  \caption{The 24 chain-squares that appear in Penrose rhombus tilings, with labels.}
\end{figure}

The key ingredients for this construction are the chains of tiles, and the presence of a uniformly recurrent tile that ensures that there are finitely many corresponding chain-squares.
Some non-parallelogram tiles, such as the Penrose Kite-and-Darts, have a combinatorial structure that satisfy the first condition. 




\bibliography{wangpotatoes}
\end{document}